\documentclass[12pt,letterpaper]{article}
\usepackage{mathptmx}  % font: Times Roman
\usepackage{hyperref}
\usepackage[dvipsnames]{xcolor}
\usepackage[margin=1in,footskip=0.5in]{geometry}
\usepackage[T1]{fontenc}
\usepackage{verbatim}
\usepackage{amsthm}
\usepackage{amsmath}
\usepackage{amssymb}
\usepackage{caption}
\usepackage{natbib}
\usepackage{bbm}
\usepackage{setspace}
\usepackage{subfigure}
\usepackage{xcolor}
\usepackage{graphics}
\usepackage{epsfig}
\usepackage{geometry}
\usepackage{accents}

\usepackage{enumitem}
\usepackage{setspace}
\usepackage{array}
\newcolumntype{x}[1]{>{\centering\arraybackslash\hspace{0pt}}p{#1}}
\usepackage{tikz}
\usetikzlibrary{decorations.pathreplacing,calligraphy}
\usetikzlibrary{patterns}
\usetikzlibrary{patterns.meta}
\pgfdeclarepattern{
  name=hatch,
  parameters={\hatchsize,\hatchangle,\hatchlinewidth},
  bottom left={\pgfpoint{-.1pt}{-.1pt}},
  top right={\pgfpoint{\hatchsize+.1pt}{\hatchsize+.1pt}},
  tile size={\pgfpoint{\hatchsize}{\hatchsize}},
  tile transformation={\pgftransformrotate{\hatchangle}},
  code={
    \pgfsetlinewidth{\hatchlinewidth}
    \pgfpathmoveto{\pgfpoint{-.1pt}{-.1pt}}
    \pgfpathlineto{\pgfpoint{\hatchsize+.1pt}{\hatchsize+.1pt}}
    \pgfpathmoveto{\pgfpoint{-.1pt}{\hatchsize+.1pt}}
    \pgfpathlineto{\pgfpoint{\hatchsize+.1pt}{-.1pt}}
    \pgfusepath{stroke}
  }
}

\tikzset{
  hatch size/.store in=\hatchsize,
  hatch angle/.store in=\hatchangle,
  hatch line width/.store in=\hatchlinewidth,
  hatch size=5pt,
  hatch angle=0pt,
  hatch line width=.5pt,
}
\usepackage{pifont}
\newcommand{\R}{\mathbb{R}}%
\newcolumntype{L}[1]{>{\raggedright\let\newline\\\arraybackslash\hspace{0pt}}m{#1}}
\newcolumntype{C}[1]{>{\centering\let\newline\\\arraybackslash\hspace{0pt}}m{#1}}
\newcolumntype{R}[1]{>{\raggedleft\let\newline\\\arraybackslash\hspace{0pt}}m{#1}}

\makeatletter \renewenvironment{proof}[1][\proofname]
{\par\pushQED{\qed}\normalfont\topsep6\p@\@plus6\p@\relax\trivlist\item[\hskip\labelsep\bfseries#1\@addpunct{.}]\ignorespaces}{\popQED\endtrivlist\@endpefalse} \makeatother

\theoremstyle{plain}
\newtheorem{thm}{Theorem}

\newtheorem{defn}{Definition}

\newtheorem{cor}{Corollary}
\newtheorem{prop}{Proposition}

\usepackage{multirow}
\usepackage{chngcntr}
\usepackage{apptools}
\AtAppendix{\counterwithin{lem}{section}}
\AtAppendix{\counterwithin{prop}{section}}

\definecolor{uncblue}{RGB}{75, 165, 211}
\definecolor{nberblue}{RGB}{0, 90, 155}
\definecolor{resgreen}{RGB}{34, 84, 67}
\definecolor{pennblue}{RGB}{0, 44, 119}
\definecolor{pennred}{RGB}{152, 30, 50}

\hypersetup{%backref,
pdfborder = {0 0 0},
urlbordercolor = {0 0 0},
colorlinks=true,
linkcolor=pennblue,
urlcolor=nberblue,%pennred,
citecolor=pennblue}

\newcounter{parentnumber}

\begin{document}

\author{Dihan Zou\thanks{Economics Department, University of North Carolina at Chapel Hill. Email: \href{mailto:dihan@email.unc.edu}{\texttt{dihan@email.unc.edu.}}}
}
\date{\today}
\title{Log-concave functions and transformations thereof\footnote{
These notes are intended for teaching, and will not be under consideration for publication. I thank Fei Li for his encouragement.  
}}

\maketitle

    \vspace{-.3in}

    \begin{abstract}
        I summarize \cite{BagBerg05Logconcave}'s review of log-concave functions, make several corrections, and augment the discussion with further results that can be useful in establishing monotone hazard rates. I also provide an application to monopoly pricing, where log-concavity of the demand curve \textit{implies} strict concavity of the revenue function in quantity.
\end{abstract}

\tableofcontents

\section{Properties of Log-Concave Functions}
Throughout, we are interested in functions mapping from a connected set $D\subseteq \mathbb{R}$ to $\mathbb{R}_{+}$. They are assumed positive and twice continuously differentiable on the open intervals where their logarithms are considered. Densities have interval support, are positive and twice continuously differentiable in its interior, and are extended by zero outside the support.

\begin{defn}\label{def:lc}
  $f(x)$ is \textit{log-concave} on an interval $(a,b) \subset D$ if \[
 (\log (f(x)))'' \leq 0, \forall x \in (a, b).
 \]
 It is \textit{strictly} log-concave on $(a,b)$ if $\log f$ is strictly concave.
\end{defn}

\noindent {\bf Remark.} $f$ is said to be a \textit{log-concave function} if $f$ is log-concave on $D$. 

The following propositions describe some properties of log-concave functions. 
The first is immediate if one writes the derivative condition in Definition \ref{def:lc} explicitly. 
\begin{prop}\label{prop:def}
    The following statements are equivalent: 
    \begin{itemize}
        \item[1.] $f(x)$ is log-concave on $(a,b)$;
        \item[2.] $\frac{f'(x)}{f(x)}$ is nonincreasing on $(a,b)$;
        \item[3.] $f'' \cdot f - (f')^2 \leq 0$ on $(a,b)$.
    \end{itemize}
\end{prop}

\begin{prop}\label{prop_concave}
  The following statements are true:
  \begin{itemize}
      \item[1.] A positive concave function is log-concave. 
      \item[2.] A log-concave function is quasi-concave.
  \end{itemize}
\end{prop}
\begin{proof}
    See Appendix \ref{apx:concave}.
\end{proof}

\begin{prop}\label{prop_unimodal}
    If $f$ is a log-concave function, then $f$ is unimodal.\footnote{
    The unimodality here allows monotone functions and boundary or limiting modes.
    }
\end{prop}
\begin{proof}
Since \(f'/f\) is nonincreasing and \(f>0\), the derivative \(f'\) can change sign only from positive to negative, possibly through an interval on which it is zero. Hence \(f\) is nondecreasing before its mode or modal interval and nonincreasing afterward.
\end{proof}

\begin{defn}[Monotone likelihood ratio property]\label{defn:mlrp}
A family of probability density functions $\{f(x\mid\theta)\}$ has the
\emph{monotone likelihood ratio property (MLRP)} if, for every
$\theta_1<\theta_2$ and $x<x'$,
\[
f(x\mid\theta_2)f(x'\mid\theta_1)
\leq
f(x\mid\theta_1)f(x'\mid\theta_2).
\]
When the densities are positive on a common support, this is equivalent
to requiring $f(x\mid\theta_2)/f(x\mid\theta_1)$ to be nondecreasing in
$x$. The cross-product formulation is used whenever a denominator is
zero.
\end{defn}

\begin{prop}\label{prop_mlrp}
A family of probability density functions $\{f(x-\theta), \theta \in \R\}$ satisfies MLRP if and only if $f$ is log-concave.
\end{prop}
\begin{proof}
    See Appendix \ref{apx:mlrp}.
\end{proof}

\noindent {\bf Remark.} For example, the family of normal distributions $\{N(\mu, \sigma^2)\}$ is commonly used by varying $\mu$. It satisfies the MLRP, and each normal density function is log-concave. In fact, one can use any log-concave density function to construct a family that satisfies MLRP, according to Proposition \ref{prop_mlrp}.

\begin{thm}[\cite{BagBerg05Logconcave}] \label{thm_begets}
Suppose that $f(x)$ is log-concave on $(a,b)$. 
Then\begin{itemize}
        \item[1.] If $F(x) := \int_a^x f(t)\,dt$ is finite for all $x \in (a,b)$, it is log-concave on $(a,b)$;
        \item[2.] If $\overline{F}(x) := \int_x^b f(t)\,dt$ is finite for all $x \in (a,b)$, it is log-concave on $(a,b)$.
    \end{itemize}
\end{thm}
\begin{proof}
First suppose that $a$ and $b$ are finite.
    We begin by proving $F(x)$ is log-concave. By Proposition \ref{prop:def}, it suffices to show that $F'' \cdot F - (F')^2 \leq 0$ on $(a,b)$, or \begin{equation}
        f' \cdot F - f^2 \leq 0\label{target_ineq}
    \end{equation}
    on $(a,b)$.
    Since $f$ is log-concave on $(a,b)$, $f'(x)/f(x)$ is nonincreasing on $(a,b)$. Then, \begin{align*}
        \frac{f'(x)\cdot F(x)}{f(x)} &=  \frac{f'(x)}{f(x)} \int_a^xf(t)\, dt  \\
            &\leq \int_a^x \frac{f'(t)}{f(t)} f(t)\, dt\\
            &= \int_a^x f'(t)\, dt = f(x) - f(a),
    \end{align*} 
    i.e., $\frac{f'(x)\cdot F(x)}{f(x)} \leq  f(x) - f(a)$.
    Rearranging yields \[
    f'\cdot F - f^2 \leq -f(a)\cdot f \leq 0,
    \]
    implying inequality \eqref{target_ineq} holds.\footnote{
    Note that the last inequality is weak because we can only guarantee $f(a) := \lim_{x\rightarrow a}f(x) \geq 0$. Later, we use $f(b) \geq 0$ for the same reason.
    } Hence, $F(x)$ is log-concave on $(a,b)$. 

    On the other hand, $\overline{F}(x)$ being log-concave requires \[
    -f' \cdot \overline{F} - (f)^2 \leq 0.
    \]
    Using the same logic, \begin{align*}
        \frac{f'(x)\cdot \overline{F}(x)}{f(x)} &=  \frac{f'(x)}{f(x)} \int_x^bf(t)\, dt  \\
            &\geq \int_x^b \frac{f'(t)}{f(t)} f(t)\, dt\\
            &= \int_x^b f'(t)\, dt = f(b) - f(x),
    \end{align*} 
    which implies \[
    -f' \cdot \overline{F} - (f)^2 \leq -f(b) \cdot f(x) \leq 0. 
    \]
    Therefore, $\overline{F}(x)$ is log-concave on $(a,b)$.

    For infinite endpoints, approximate $(a,b)$ by finite
intervals $(a_n,b_n)$ with $a_n\downarrow a$ and
$b_n\uparrow b$. The corresponding integrals
\[
    F_n(x):=\int_{a_n}^x f(t)\,dt,
    \qquad
    \overline{F}_n(x):=\int_x^{b_n} f(t)\,dt
\]
are log-concave on $(a_n,b_n)$. Whenever the respective
limiting integral is finite, the log-concavity inequality
passes to the pointwise limit.
\end{proof}

\begin{cor}
    The \emph{cumulative reliability function} $H(x) = \int_x^b \overline{F}(t)\, dt$ is log-concave on $(a,b)$ if $\overline{F}(x)$ is log-concave on $(a,b)$.
\end{cor}
\begin{cor}
    The \emph{mean residual lifetime function} $\frac{\int_x^b \overline{F}(t)\, dt}{\overline{F}(x)} = -\frac{H(x)}{H'(x)} $ is nonincreasing on $(a,b)$ if $\overline{F}(x)$ is log-concave on $(a,b)$.
\end{cor}
The corollaries immediately follow from Theorem \ref{thm_begets}. Note that a stronger sufficient condition for both corollaries to hold is $f(x)$, as opposed to $\overline{F}(x)$, being log-concave on $(a,b)$. The following corollary clarifies why $f$ being log-concave is a stronger condition. 
With a finite upper endpoint, log-concavity of the density
also yields a strict monotonicity property.

\begin{cor}
    \label{cor:strict-inverse-hazard}
    Let $F$ have a positive, log-concave density $f$ on its
    support $(a,b)$, where $b<\infty$.\footnote{
    Without a finite upper support endpoint,
    strictness need not hold. Indeed, an exponential distribution
    with rate $\lambda$ has $M(x)=1/\lambda$.
    } Then
    \[
        M(x):=\frac{1-F(x)}{f(x)}
    \]
    satisfies $M'(x)<0$ for all $x\in(a,b)$.
\end{cor}

\begin{proof}
    Fix $x\in(a,b)$ and write $s:=(\log f)'(x)$.
    Differentiating gives
    \[
        M'(x)=-1-sM(x).
    \]
    If $s\geq0$, then $M'(x)<0$ immediately. If $s<0$,
    concavity of $\log f$ gives
    $\log f(t)\leq\log f(x)+s(t-x)$ for $x\leq t<b$.
    Hence
    \[
        M(x)
        =\int_x^b\frac{f(t)}{f(x)}\,dt
        \leq\int_x^b e^{s(t-x)}\,dt
        <\int_x^\infty e^{s(t-x)}\,dt
        =-\frac1s.
    \]
    Thus $M'(x)=-1-sM(x)<0$ in this case as well.
\end{proof}

\section{Other Transformations}
In the previous section, I have shown that log-concavity is preserved under integration. I show other transformations that preserve log-concavity in this section.
\subsection{Multiplication}
\begin{thm}
    If $f$ and $g$ are log-concave, then $h(x) = f(x)\cdot g(x)$ is log-concave.
\end{thm}
\begin{proof}
    Since concavity is preserved under sums, $\log h(x) = \log f(x) + \log g(x)$ is concave in $x$, i.e., $h$ is log-concave.
\end{proof}

\noindent {\bf Remark.} Summation does not preserve log-concavity in general. To gain intuition, consider two unimodal density functions with different modes. Their sum may not be unimodal anymore.

\pagebreak

\subsection{Composition}

\begin{thm}\label{thm:composition}
    Suppose $f$ is log-concave on $(a,b)$. 
    Let $t$ be a twice-differentiable function mapping from $(a',b')$ to $t((a',b')) \subseteq (a,b)$. Define the function $f_t$ with support $(a',b')$ so that for all $x \in (a', b')$, $f_t(x) = f(t(x))$. 
    \begin{itemize}
        \item[(i)] If $f$ is increasing and $t$ is concave, then $f_t$ is log-concave.\footnote{
        \cite{BagBerg05Logconcave} implicitly assume $f$ is increasing and impose a redundant requirement that $t$ is monotone in their related results (Theorem 7). \cite{ZouETBComment} discusses the problem in detail, and proposes a new version of the theorem. 
        }
        
        \item[(ii)] If $f$ is decreasing and $t$ is convex, then $f_t$ is log-concave.
    \end{itemize}
\end{thm}
\begin{proof}
    Examining the second-order derivative on $\log f_t$ with respect to $x$, we have \[
    \frac{(t')^2[f''\cdot f - (f')^2 ] + f\cdot f' \cdot t''}{f^2}.
    \]
    A sufficient condition for the above expression to be nonpositive (so that $f_t$ is log-concave) is that $f\cdot f' \cdot t'' \leq 0$, because the first term in the numerator is nonpositive by log-concavity of $f$ (c.f. Proposition \ref{prop:def}). Therefore, either (i) or (ii) is a sufficient condition for $f_t$ to be log-concave.
\end{proof}

For an affine transformation, $t''=0$, so the composition
preserves log-concavity without any monotonicity assumption
on the outer function. 
\begin{cor}
    Suppose $f$ is log-concave on $(a,b)$. Let $t$ be an affine transformation from $\R$ to $\R$. Define a function $f_t = f\circ t$ with support $t^{-1}((a,b))$, then $f_t$ is log-concave on $t^{-1}((a,b))$.
\end{cor}
\begin{proof}
    Apply the derivative formula in the proof of Theorem \ref{thm:composition} with \(t''=0\).
\end{proof}

It is sometimes useful to set the transformation $t$ to be the ratio $\Gamma(x) : = \frac{F(x)}{f(x)}$. For uniform distributions, $\Gamma(x)$ is affine. 
For exponential distributions, $\Gamma(x)$ is convex, so $G \circ \Gamma$ is log-concave if some function $G$ is log-concave and \textit{decreasing}. The following proposition shows, which is less immediate, that $\Gamma(x)$ is also convex for normal distributions.

\begin{prop}
    Let $F$ and $f$ be the cumulative distribution function
    and density of $N(\mu,\sigma^2)$, where $\mu\in\mathbb{R}$
    and $\sigma>0$. Then $\Gamma(x):=F(x)/f(x)$ is strictly
    convex on $\mathbb{R}$.
\end{prop}

\begin{proof}
    Changing variables $u=x-t$ gives
    \[
        \Gamma(x)
        =\int_{-\infty}^x \frac{f(t)}{f(x)}\,dt
        =\int_0^\infty
        \exp\!\left(
            \frac{(x-\mu)u}{\sigma^2}
            -\frac{u^2}{2\sigma^2}
        \right)\,du.
    \]
    Gaussian tail decay justifies differentiating twice
    under the integral sign, locally uniformly in $x$. Hence
    \[
        \Gamma''(x)
        =\frac{1}{\sigma^4}\int_0^\infty u^2
        \exp\!\left(
            \frac{(x-\mu)u}{\sigma^2}
            -\frac{u^2}{2\sigma^2}
        \right)\,du
        >0.
    \]
\end{proof}

\subsection{Truncation} 

\begin{thm}\label{thm_truncation}
    If a probability distribution has a log-concave density function (resp. cumulative distribution function), then any truncation of this probability distribution to an interval of positive probability will also have a log-concave density function (resp. cumulative distribution function).
\end{thm}
\begin{proof}
    Let $F$ and $f$ denote the original CDF and density.
    Truncate to an interval $(a,b)$ with
    $F(b)-F(a)>0$.
    On the truncated support,
    \[
        \widetilde f(x)=\frac{f(x)}{F(b)-F(a)},
        \qquad
        \widetilde F(x)=\frac{F(x)-F(a)}{F(b)-F(a)}.
    \]
    If $f$ is log-concave, then
    $\log\widetilde f=\log f-\log (F(b)-F(a))$ is concave.

    If $F$ is log-concave, define
    $h(z):=\log(e^z-F(a))$ wherever $e^z>F(a)$. Since
    \[
        h'(z)=\frac{e^z}{e^z-F(a)}>0,
        \qquad
        h''(z)=-\frac{F(a)e^z}{(e^z-F(a))^2}\leq0,
    \]
    $h$ is increasing and concave. Consequently,
    \[
        \log\widetilde F(x)
        =h(\log F(x))-\log (F(b)-F(a))
    \]
    is concave, proving the CDF claim.
\end{proof}

The following example of truncated normal distributions reiterates the idea of Theorem \ref{thm_truncation}.\footnote{
The corresponding result in \cite{BagBerg05Logconcave} (Theorem 9) 
overstates the preservation of \textit{log-convexity}: lower truncation need not preserve log-convexity of the CDF, although log-convexity of the density on the truncated support is preserved.
Nevertheless, see Section 6 of \cite{BagBerg05Logconcave} for a nice and thorough discussion on examples of log-concave distributions. 
}

\paragraph{Truncated Normal Distributions.}
Denote the truncated normal distribution by  $\texttt{trunc}_{[a,b]}N(\mu, \sigma^2)$ for shorthand, meaning that the distribution is $N(\mu, \sigma^2)$ conditional on  the draw lying in $[a,b]$.
Then, $\texttt{trunc}_{[a,b]}N(\mu, \sigma^2)$ has a conditional cumulative distribution function defined as\footnote{
See, e.g., \cite{Greene2002}, pp.757-759.
} \[
F(x) = \frac{\Phi(\xi) - \Phi(\alpha)}{\Phi(\beta) - \Phi(\alpha)},\; x \in [a,b]
\]
with $\Phi(\cdot)$ being the cdf of the standard normal distribution; $\alpha := \frac{a - \mu}{\sigma}$ is the standardized lower bound; $\beta := \frac{b - \mu}{\sigma}$ is the standardized upper bound; $\xi:= \frac{x-\mu}{\sigma}$ is the standardized statistic.
The probability density function is \[
    f(x) = \frac{\varphi(\xi)}{\sigma (\Phi(\beta) - \Phi(\alpha))},\;  x \in [a,b]
\]
Therefore, $\texttt{trunc}_{[0,1]}N(\mu, \sigma^2)$ has a cdf of \[
F(x;\mu,\sigma) = \frac{\Phi(\frac{x-\mu}{\sigma}) - (1- \Phi(\frac{\mu}{\sigma}))}{\Phi(\frac{1-\mu}{\sigma})+ \Phi(\frac{\mu}{\sigma}) - 1},\; x \in [0,1]
\]
and a pdf of \[
f(x;\mu,\sigma) = \frac{\varphi(\frac{x -\mu}{\sigma})}{\sigma[\Phi(\frac{1-\mu}{\sigma})+ \Phi(\frac{\mu}{\sigma}) - 1]},\; x \in [0,1]
\]
Truncation preserves log-concavity of both the normal density and its CDF.

\smallskip

\noindent{\it Uniform Limit as $\sigma \rightarrow \infty$.}
For any fixed $\mu\in\mathbb{R}$, the density and cumulative
distribution function of $\operatorname{trunc}_{[0,1]}N(\mu,\sigma^2)$
converge uniformly on $[0,1]$ to those of the uniform distribution.
To see this, write the density as
\[
    f(x;\mu,\sigma)
    =
    \frac{\exp\!\left(-\frac{(x-\mu)^2}{2\sigma^2}\right)}
    {\displaystyle\int_0^1
    \exp\!\left(-\frac{(t-\mu)^2}{2\sigma^2}\right)\,dt},
    \qquad x\in[0,1].
\]
Since $\mu$ is fixed and $[0,1]$ is bounded, the numerator
converges uniformly to $1$, while the denominator converges
to $1$. Consequently,
\[
    \sup_{x\in[0,1]}|f(x;\mu,\sigma)-1|
    \xrightarrow{\sigma\to\infty}0.
\]
Moreover, $F(x;\mu,\sigma)=\int_0^x f(t;\mu,\sigma)\,dt$ implies
\[
    \sup_{x\in[0,1]}|F(x;\mu,\sigma)-x|
    \leq
    \sup_{t\in[0,1]}|f(t;\mu,\sigma)-1|
    \xrightarrow{\sigma\to\infty}0.
\]

\section{Application: Log-Concave Demand Curve}
In this section, I provide an application to monopoly pricing to illustrate the role that log-concave distributions play in economic analysis. 

There is a single product market. Consider a continuum of consumers, each of whom has one unit of indivisible demand of the product. Each consumer's value of the unit of demand is $v \in [0,1]$ drawn i.i.d.\ from a distribution $G$ with log-concave, full-support density $g$.\footnote{
A weaker condition is to assume $1-G(v)$ is strictly log-concave, with which the
revenue-concavity and strict monotonicity in Proposition \ref{prop_concave_rev} go through.
} 

There is a monopoly firm serving the product.\footnote{%
One could immediately tell from a mechanism design perspective, that $g$ or $1 - G$ being log-concave implies the virtual value $\psi(v) := v - \frac{1-G(v)}{g(v)}$ is increasing. Then, the profit maximizing mechanism is a posted price mechanism with $p = \psi^{-1}(c)$ (see, e.g., \cite{borgers2015}).
In this section, I deliberately use the more classical demand curve approach to illustrate the relationship between log-concave value distribution and the property on the demand curve (\cite{BulowRoberts89} first formally point out the link between the two).  
} The firm posts a unit price $p \geq 0$. Then the market demand is $q(p) = \Pr(v \geq p) = 1 - G(p)$. 
The monopolist solves \[
\max_p (p-c)q(p),
\]
where $c \in [0,1)$ is the firm's marginal cost of production. Assume no fixed cost. 
The following proposition says strict log-concavity of the demand function is a stronger property than the \textit{revenue being strictly concave in $q$}.\footnote{
More commonly seen in the IO and regulation literature. See, e.g., Assumption 1 and Proposition 2 of \cite{WeiZouRegulationNoSubsisdy}. In fact, log-concavity of demand implies strict concavity of profit in quantity \textit{and} strict log-concavity of profit in price on $(c,1)$, either of which guarantees a unique monopoly solution. Ordinary concavity in price requires a separate condition, such as \(2g(p)+(p-c)g'(p)>0\) over $p \in (c, 1)$.
}
Moreover, the optimal price has a decreasing markup as cost goes up.

\begin{prop}\label{prop_concave_rev}
    The revenue function of the monopoly problem is strictly concave in quantity $q$. 
    Moreover, the optimal price is the unique solution
$p\in(c,1)$ to \[
    p = c + \frac{1 - G(p)}{g(p)},
    \]
    with the markup $\frac{1 - G(p)}{g(p)}$ strictly decreasing as $c$ increases. 
\end{prop}
\begin{proof}
    First, we prove that the monopolist's revenue is strictly concave in $q$.  $G$ is strictly increasing since it is full support. Then, the inverse of $G$ is well-defined.
    Rewriting the monopolist's problem by substituting $p(q) = G^{-1}(1-q)$, we get \[
    \max_q [G^{-1}(1-q) - c]q
    \]
    Taking derivative w.r.t.\ $q$, we get \[
    \underbrace{p -\frac{1 - G(p)}{g(p)}}_{\text{marginal revenue}}  \quad - \quad c.
    \]
    Note that the marginal revenue is strictly increasing in $p$ because the term $\frac{1-G(p)}{g(p)}$ is decreasing by log-concavity of $g$. Meanwhile, $p$ is strictly decreasing with $q$, which implies the marginal revenue $p(q) - \frac{1-G(p(q))}{g(p(q))}$ is strictly decreasing in $q$. That is, the revenue function is strictly concave in $q$. 

    Profit is continuous in $p$ on $[c,1]$, zero at both endpoints,
and strictly positive on $(c,1)$; prices outside $(c,1)$
yield nonpositive profit. Hence an interior maximizer exists.
Strict concavity in quantity implies uniqueness, so the
first-order condition characterizes the optimal price:\[
    p(c) = c+ \frac{1 - G(p(c))}{g(p(c))}.
    \]
    One can reinterpret the term $\frac{1 - G(p)}{g(p)}$ as the monopoly markup. The implicit function theorem implies $p$ increases as $c$ increases,  which implies the markup $\frac{1 - G(p)}{g(p)}$ decreases in $c$. 
   Corollary \ref{cor:strict-inverse-hazard} implies the markup is strictly decreasing. 
\end{proof}
 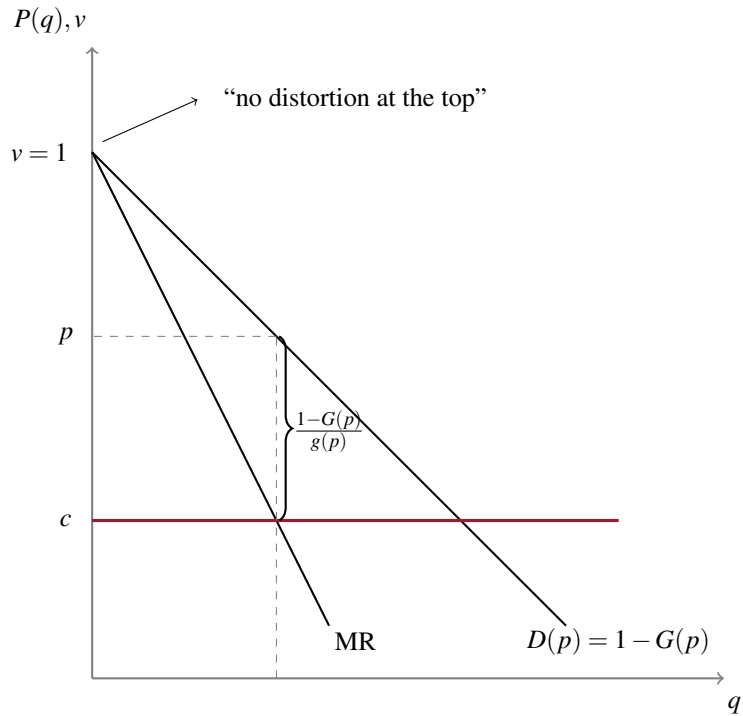
\begin{figure}[htbp]
    \centering 
        \begin{tikzpicture}[scale=7]
        \footnotesize
            %  %% %% %% Left Panel %% %% %% %% %% %
            \draw [->, thick,gray] (0,0) -- (0,1.2);
            \draw [->, thick,gray] (0,0) -- (1.2,0);
           
            \node at (-0.08, 1.25) {$P(q), v$};

            \node at (1.22, -0.05) {$q$};

            \draw[-, thick] (0,1) -- (0.9,0.1);
            \node at (1, 0.07) {$P(q) = G^{-1}(1-q)$};

            \draw[-, thick] (0,1) -- (0.45, 0.1);
            \node at (0.5, 0.07) {MR};

            \draw[-, pennred, very thick] (0,0.3) -- (1, 0.3);
            \node at (-0.05, 0.3) {$c$};

            \draw[-, black!50, dashed] (0.35, 0) -- (0.35, 0.65);

            \draw[-, black!50, dashed] (0.35, 0.65) -- (0, 0.65);
            \node at (-0.05, 0.65) {$p$};

            \draw[->] (0.02, 1.02) -- (0.2, 1.1);
            \node at (0.5, 1.1) {``no distortion at the top''};

            \draw [decorate, decoration = {brace,mirror, raise = 1pt,amplitude=5pt}, thick] (0.35,0.3) --  (0.35,0.65);

            \node at (0.45, 0.47) {$\frac{1-G(p)}{g(p)}$};

            \node at (-0.1,1) {$v = 1$};

        \end{tikzpicture}

        \caption{Illustration of decreasing markup}
        \label{fig_mr}
    \end{figure}

\noindent{\bf Remark.}
The fact that the price markup decreases in cost $c$ is a distinctive feature of log-concave demand. This property does not hold for general demand systems with concave revenue; for example, constant-elasticity demand (with elasticity larger than $1$)  exhibits increasing markup.

The intuition is most transparent from the perspective of information rents. For a log-concave density $g$, the information rent at type $v$, given by\[
\frac{1-G(v)}{g(v)},
\]
is strictly decreasing in 
$v$. Serving any type incurs incentive externalities on all higher types, so the marginal benefit of serving consumers rises more steeply than in the first best. As 
$v$ approaches its upper bound, these incentive effects disappear—hence “no distortion at the top’’ (see Figure \ref{fig_mr}).

Under the first-best benchmark, the planner serves consumers until their value equals cost: $\hat{v}^{FB} = c$. A monopolist instead sets output by equating marginal revenue and cost:\[
\text{MR}(\hat{v})= \hat{v} - \frac{1 - G(\hat{v})}{g(\hat{v})} = c,
\]
leading to under-provision. Implementing this through a unit price requires\[
p - \frac{1-G(p)}{g(p)} = c,
\]
so the indifferent consumer has value $\hat{v}=p$. Rearranging gives \[
p(c) = c + \frac{1 - G(\hat{v})}{g(\hat{v})}, \hat{v} = p(c),
\]
revealing that the price markup equals the information rent at the marginal (indifferent) consumer.\footnote{
It seems oxymoronic, but the information rent at the indifferent type is actually the cumulative rent that the monopolist firm leaves for the \textit{types higher than} the indifferent type.
}
Because the monopolist internalizes the incentive externality of extending trade, its pricing rule embeds exactly this rent—capturing the misalignment between monopolist and planner.

As cost $c$ increases, the indifferent consumer must lie at a higher value. Since information rents decrease in $v$, the markup necessarily falls. Intuitively, the monopolist’s objective becomes more \textit{aligned} with the social planner when serving higher-valued consumers, explaining why markups decrease with cost.

\smallskip

\noindent{\it Demand-side interpretation.}
The same phenomenon can be understood directly from the
shape of the demand curve. Under log-concavity of $g$,
higher prices place the monopolist on a more price-sensitive
segment of demand, in the following precise sense.

\begin{prop}
    \label{prop:elasticities}
 Define the price
    \emph{semi-elasticity} and \emph{elasticity} of demand $q(p)=1-G(p)$,
    both in absolute value, by
    \[
        h(p):=-\frac{q'(p)}{q(p)}
        =\frac{g(p)}{1-G(p)}
        \text{\quad and \quad }
        \eta(p):=-\frac{pq'(p)}{q(p)}=ph(p).
    \]
    Both $h$ and $\eta$ are strictly increasing on $(0,1)$.
\end{prop}

\begin{proof}
    Let $M(p):=[1-G(p)]/g(p)$. By Corollary \ref{cor:strict-inverse-hazard},
    $M'(p)<0$. Since $h=1/M>0$,
    \[
        h'(p)=-\frac{M'(p)}{M(p)^2}>0,
        \qquad
        \eta'(p)=h(p)+ph'(p)>0
    \]
    for every $p\in(0,1)$.
\end{proof}

To understand the pricing implication, consider the firm's
tradeoff when expanding output. Selling an additional unit
requires lowering the price on all existing sales. Marginal
revenue therefore equals the price of the additional unit
minus the revenue sacrificed on existing sales:
\[
    MR(q)=P(q)+qP'(q)
         =P(q)-\frac{1}{h(P(q))}.
\]
The gap between inverse demand and marginal revenue,
$-qP'(q)$, measures this revenue sacrifice per additional
unit sold. When demand is more responsive, a given
proportional expansion of sales requires a smaller price
reduction, making this gap smaller.

At the optimum, marginal revenue equals marginal cost,
so the gap is precisely the absolute markup. Equivalently,
\[
    p(c)-c=\frac{1}{h(p(c))},
    \qquad
    \frac{p(c)-c}{p(c)}=\frac{1}{\eta(p(c))}.
\]
A rise in marginal cost shifts the intersection of marginal
revenue and cost to a lower quantity and hence a higher
price. The monopolist thus moves into a more price-sensitive
region of demand: both semi-elasticity and elasticity are
higher. Consequently, the firm optimally charges a smaller
absolute markup and a smaller proportional markup, despite
charging a higher price. Geometrically, the gap between
inverse demand and marginal revenue narrows as the optimal
quantity falls (see Figure~1). This is why the price rises
by less than the increase in marginal cost.

\smallskip

\noindent{\it Market power, elasticity, and information rent.}
Price markup is commonly used to measure market power: lower elasticity implies higher markup. Yet elasticity also mirrors the information environment. A less elastic (more inelastic) demand curve corresponds to a larger pool of high-value and relatively homogeneous consumers, which in turn implies a larger information rent for the indifferent consumer. Paradoxically, a \textit{higher} information rent grants market power for the monopolist, even though in mechanism design it represents “difficulty’’ for the designer. The appropriate interpretation of this difficulty is thus the reluctance of supplying more good (and lowering price) when making the marginal tradeoff. Therefore, demand elasticity and information rent are two sides of the same coin. Both summarize the monopolist’s ability to extract surplus and thus provide dual interpretations of market power.

\medskip
\noindent{\bf Remark.} (A weaker condition for strictly concave revenue) 
Finally, we end this section by identifying a notable weaker condition for the revenue function to be strictly concave in $q$. For the following proposition, the assumption of log-concavity of $g$ is dropped. Other assumptions of the setting are maintained.

\begin{prop}
    \label{prop:logconcave-demand}
    If the demand $D(p):=1-G(p)$ is (weakly) log-concave,
    then the revenue $qP(q)$
    is strictly concave.
\end{prop}

\begin{proof}
    Define $M(p):=D(p)/g(p)$. Since
    $(\log D)'=-1/M$, log-concavity of $D$ implies
    \[
        M'(p)=M(p)^2(\log D)''(p)\leq0.
    \]
    Moreover, $P'(q)=-1/g(P(q))$ and $D(P(q))=q$, so
    \[
        (qP(q))'=P(q)-M(P(q)).
    \]
    Differentiating again yields
    \[
        (qP(q))''
        =-\frac{1-M'(P(q))}{g(P(q))}
        \leq-\frac{1}{g(P(q))}<0.
    \]
    The revenue is thus strictly concave.
\end{proof}

\section{Related Literature}

Log-concave functions and distributions have been studied extensively in probability, convex geometry, and stochastic programming. Foundational work by \cite{Prekopa1973} introduced logarithmically concave measures and established key preservation properties (e.g., under marginals and integrals) that underpin modern treatments of log-concavity and the Prékopa--Leindler inequality. More systematic probability-theoretic accounts of unimodality and convexity, including log-concavity and its links to reliability theory and monotone hazard rates, are provided in \cite{DharmadhikariJoagDev1988} and surveyed in the statistics literature by \cite{SaumardWellner2014}.

Within economics, \cite{An1998} provides a complete characterization of
log-concavity without imposing differentiability on density functions,
systematically compares log-concavity with log-convexity, identifies
their different behavior under one-sided integration as a key source of
asymmetry, and discusses multivariate extensions.
The review by \cite{BagBerg05Logconcave} synthesizes many of these ideas for log-concave probabilities and highlights applications in appraisal design in lemon markets. The properties documented here---preservation of log-concavity under integration, truncation, and economically natural transformations---are closely related to conditions used in mechanism design and contract theory; see, for example, \cite{borgers2015} for mechanism design, where monotone hazard rate plays a technical role; or \cite{BoltonConcract} for contract theory, where MLRP helps model the monitoring technology.

\appendix
\section{Omitted Proofs}
\subsection{Proof of Proposition \ref{prop_concave}} \label{apx:concave}
{\it Proof of 1.} Suppose $f$ is a positive concave function. By definition, for any $x, y$ in the domain of $f$, any $\alpha \in [0,1]$, \[
f(\alpha x + (1-\alpha) y) \geq \alpha f(x) + (1-\alpha) f(y).
\]
Taking logarithms on both sides, we get
\[
\log f(\alpha x + (1-\alpha) y) \geq \log (\alpha f(x) + (1-\alpha) f(y)) \geq \alpha \log f(x) + (1-\alpha) \log f(y), 
\]
where the last inequality is implied by concavity of $\log(\cdot)$. Therefore, $f$ is log-concave.

\smallskip

\noindent {\it Proof of 2.} Suppose $f$ is log-concave, we have, for any $x,y$ in the domain of $f$ and any $\alpha \in [0,1]$, \[
\log f(\alpha x + (1-\alpha) y) \geq \alpha \log f(x) + (1-\alpha) \log f(y), 
\]
the right-hand side of which is greater than or equal to $\min\{\log f(x), \log f(y)\}$. Therefore $\log f$ is quasi-concave. Note that the upper level sets of $f$, defined as $\{x|f(x) \geq a\}$ for $a > 0$, are equivalent to $\{x|\log f(x) \geq \log a\}$, i.e., the upper level sets of $\log f$. Since $\log f$ is quasi-concave, its upper level sets are convex. Therefore, the upper level sets of $f$ are convex, confirming that $f$ is quasi-concave.  

\subsection{Proof of Proposition \ref{prop_mlrp}}
\label{apx:mlrp}

\noindent{\it $f$ is log-concave $\Rightarrow$ MLRP.}
Fix $\theta_1<\theta_2$ and $x<x'$, and define
\[
a=x-\theta_2,\qquad b=x'-\theta_1,\qquad
t=\frac{x'-x}{\theta_2-\theta_1+x'-x}\in(0,1).
\]
Then
\[
x-\theta_1=ta+(1-t)b,
\qquad
x'-\theta_2=(1-t)a+tb.
\]
If $f(a)f(b)=0$, the cross-product inequality holds immediately.
Otherwise, $a$ and $b$ belong to the support of $f$. Since this support
is an interval, the two convex combinations above also belong to it.
Concavity of $\log f$ therefore gives
\[
\log f(x-\theta_1)
\geq t\log f(a)+(1-t)\log f(b)
\]
and
\[
\log f(x'-\theta_2)
\geq (1-t)\log f(a)+t\log f(b).
\]
Adding these inequalities and exponentiating yields
\[
f(x-\theta_2)f(x'-\theta_1)
\leq
f(x-\theta_1)f(x'-\theta_2),
\]
which is MLRP.

\smallskip

\noindent{\it MLRP $\Rightarrow$ $f$ is log-concave.}
Let $a<b$ lie in the support of $f$, and set
\[
m=\frac{a+b}{2},\qquad
\theta_1=0,\qquad
\theta_2=\frac{b-a}{2},\qquad
x=m,\qquad x'=b.
\]
Then
\[
x-\theta_2=a,\qquad x'-\theta_1=b,\qquad
x-\theta_1=x'-\theta_2=m.
\]
MLRP consequently implies
\[
f(a)f(b)\leq f(m)^2,
\]
or
\[
\log f\left(\frac{a+b}{2}\right)
\geq
\frac{\log f(a)+\log f(b)}{2}.
\]
Thus $\log f$ is midpoint concave. Under the maintained assumptions,
$f$ is positive and continuous on its interval support, so $\log f$ is
continuous there. Since a continuous midpoint-concave function is
concave, $f$ is log-concave.

\bibliographystyle{ecta}
\bibliography{references}
\end{document}